\documentclass[12pt]{article}

\usepackage[utf8]{inputenc}
\usepackage[T1]{fontenc}
\usepackage[a4paper,margin=1in]{geometry}
\usepackage{amsmath,amssymb,amsthm,mathtools}
\usepackage{aliascnt}
\usepackage{enumitem}
\usepackage{booktabs}
\usepackage[dvipsnames]{xcolor}
\usepackage{url}
\usepackage{hyperref}
\usepackage{cleveref}
\usepackage{authblk}
\usepackage{microtype}
\usepackage[backend=biber,style=alphabetic,doi=false,maxbibnames=99]{biblatex}
\mathtoolsset{showmanualtags}
\numberwithin{equation}{section}
\AtBeginDocument{%
  \setlength{\abovedisplayskip}{4pt plus 1pt minus 1pt}%
  \setlength{\belowdisplayskip}{4pt plus 1pt minus 1pt}%
  \setlength{\abovedisplayshortskip}{3pt plus 1pt minus 1pt}%
  \setlength{\belowdisplayshortskip}{3pt plus 1pt minus 1pt}%
}

\hypersetup{colorlinks=true,linkcolor=orange!80!black,urlcolor=blue,citecolor=blue}
\theoremstyle{definition}
\newtheorem{definition}{Definition}[section]
\newaliascnt{proposition}{definition}
\newtheorem{proposition}[proposition]{Proposition}
\aliascntresetthe{proposition}
\newaliascnt{theorem}{definition}
\newtheorem{theorem}[theorem]{Theorem}
\aliascntresetthe{theorem}
\newaliascnt{lemma}{definition}
\newtheorem{lemma}[lemma]{Lemma}
\aliascntresetthe{lemma}
\newaliascnt{corollary}{definition}

\aliascntresetthe{corollary}

\newtheoremstyle{italicremark}
  {.5\topsep}{.5\topsep}{\normalfont}{}
  {\itshape}{.}{.5em}
  {\thmname{#1}\thmnumber{ \textit{#2}}\thmnote{ \normalfont(#3)}}

\theoremstyle{italicremark}
\newtheorem{remark}[definition]{Remark}
\newtheorem{construction}[definition]{Construction}

\theoremstyle{remark}

\crefname{definition}{Definition}{Definitions}
\crefname{proposition}{Proposition}{Propositions}
\crefname{theorem}{Theorem}{Theorems}
\crefname{lemma}{Lemma}{Lemmas}
\crefname{corollary}{Corollary}{Corollaries}
\crefname{remark}{Remark}{Remarks}
\Crefname{definition}{Definition}{Definitions}
\Crefname{proposition}{Proposition}{Propositions}
\Crefname{theorem}{Theorem}{Theorems}
\Crefname{lemma}{Lemma}{Lemmas}
\Crefname{corollary}{Corollary}{Corollaries}
\Crefname{remark}{Remark}{Remarks}

\newcommand{\C}{\mathbb C}
\newcommand{\Z}{\mathbb Z}
\let\P\relax
\newcommand{\P}{\mathbb P}
\newcommand{\Sym}{\operatorname{Sym}}

\newcommand{\wt}{\widetilde}

\newcommand{\cO}{\mathcal O}
\newcommand{\cS}{\mathcal S}
\newcommand{\pvar}[1]{p_{#1}}
\newcommand{\qvar}[1]{q_{#1}}

\title{Segre determinants and MHV gravity amplitudes}
\author[1]{Joris Koefler\textsuperscript{*}}
\author[2]{Emanuele Ventura}
\author[2]{Hongmiao Yu}
\affil[1]{\small School of Natural Sciences, Institute for Advanced Study, 1 Einstein Drive, Princeton, USA}
\affil[2]{\small Politecnico di Torino, Dipartimento di Scienze Matematiche ``G. L. Lagrange'', Corso Duca degli Abruzzi 24, Torino, Italy}

\date{}

\begin{document}
\maketitle
\begingroup
\renewcommand{\thefootnote}{*}
\footnotetext{Email addresses: \url{koefler@ias.edu}; \url{emanuele.ventura@polito.it}; \url{hongmiao.yu@polito.it}.}
\endgroup

\begin{abstract}
We describe a {\it Segre--Veronese} construction of polynomial solutions to the interpolation conditions of the numerators of tree-level MHV gravity amplitudes given by Hodges' formula.
The five- and six-point Hodges numerators arise from Segre determinants, while for every $n\geq7$ suitable Segre determinants lead to solutions linearly independent of the Hodges numerator.
The existence of such independent elements for $n=7,8$ was recently established by Mai and Zhang.
\end{abstract}


\section{Introduction}\label{sec:intro}

The maximally helicity violating (MHV) amplitude for gravitons in $\mathcal N=8$ supergravity admits a particularly simple expression in spinor-helicity variables due to Hodges \cite[Eq.~(5)]{Hodges2012}.  An algebraic characterization of this amplitude was proposed in \cite{Gravity}.  Its formulation uses the coordinate ring of the spinor-helicity variety introduced in \cite{SH} and asks whether the degree and pairwise interpolation conditions determine the MHV numerator uniquely.

More precisely, let $Q_n$ be the spinor-helicity coordinate ring, let $J_n$ be the intersection of the ideals generated by $p_{ij}$ and $q_{ij}$ in $Q_n$, and let
\begin{equation}\label{eq:dndef-intro}
 d(n)=\left(\frac{n^2-3n-6}{2},\ n-3;\ (n-5)^n\right)\in \mathbb Z^{n+2}
\end{equation}
be the spinor multidegree of the MHV numerator, where $(n-5)^n$ means that $(n-5)$ is repeated $n$ times.  The conjecture of \cite[Conjecture~2.1]{Gravity} asserts
\[
 \dim_{\C}(J_n)_{d(n)}=1.
\]
This is also Problem 16 of \cite{problemlist}, where Paranjape first formulated the same question.
We call the numerator stemming from Hodges' formula the \emph{Hodges numerator} $N_n$. Since this is a nonzero element of $(J_n)_{d(n)}$, the conjecture is equivalent to saying that, up to a scalar, it is the unique polynomial of spinor multidegree $d(n)$ satisfying all the interpolation conditions.
The conjecture was verified at five and six points in \cite{Gravity}.  Mai and Zhang recently showed that algebraic uniqueness fails at seven and eight points and studied the additional role of permutation symmetry and physical boundary data \cite{MaiZhang}, leading to a physical uniqueness result. Related questions concerning explicit formulas and geometric structures for gravity amplitudes appear, for example, in \cite{Gravituhedron}.

The aim of this paper is to isolate a geometric mechanism behind these spaces of solutions.
The basic object is the \emph{Segre determinant} $\Delta_{a,b}$ \cite{Segredeterminant}, which is defined as the determinant of the square \emph{Segre matrix} whose rows are homogeneous coordinate vectors of $(a+1)(b+1)$ points under the Segre--Veronese embedding
\[
 \P^1\times\P^1\hookrightarrow\P^{(a+1)(b+1)-1}
\]
defined by $|\cO_{\P^1\times\P^1}(a,b)|$.
Its {\it determinantal} geometry has two features that fit the MHV problem.  First, the Segre determinant is an invariant polynomial in the spinor variables and therefore becomes a polynomial in the variables $p_{ij},q_{ij}$.  Second, whenever two labels $i,j$ satisfy $p_{ij}=q_{ij}=0$, the corresponding two rows of the Segre matrix become dependent; consequently the determinant belongs to the ideal $(p_{ij},q_{ij})$.  The spinor multidegree can then be adjusted by multiplying by suitable monomials.

This point of view already explains the two low-multiplicity cases where uniqueness is known.  At six points, one has $d(6)=(6,3;1^6)$,
which is exactly the spinor multidegree of the Segre determinant $\Delta_{2,1}$ associated with $\cO_{\P^1\times\P^1}(2,1)$.  Since $\Delta_{2,1}$ is nonzero and satisfies all six-point boundary conditions, the uniqueness result of \cite{Gravity} implies that it is proportional to the six-point Hodges numerator.  At five points, $d(5)=(2,2;0^5)$;
this is the degree of the determinant $\Delta_{1,1}$ on any four of the five labels.  Momentum conservation makes the five $4\times4$ Segre determinants equal up to sign in $Q_5$; hence their common class up to sign satisfies every five-point boundary condition.  It is nonzero, so five-point uniqueness again identifies it with the Hodges numerator up to scale; see \Cref{rem:five-six}.

Starting at seven points, the same geometry produces additional elements rather than recovering the unique Hodges line, generated by $N_n$. Our main result gives, for every $n\geq7$, an explicit nonzero polynomial $F_n=\Delta_n P_n\in(J_n)_{d(n)}$
that is not proportional to the Hodges numerator.  For even multiplicity the Segre determinant may be supported on all labels; for odd multiplicity at least nine we use a determinant on $n-3$ labels together with a simple three-label completion; the seven-point case is a sort of one-point extension of $\Delta_{2,1}$ in which momentum conservation supplies the missing boundary conditions. Our main theorem reads as follows.

\begin{theorem}[Segre construction at all multiplicities]\label{thm:main}
For every $n\geq7$ there exists a nonzero polynomial class $F_n=\Delta_n P_n\in(J_n)_{d(n)}$
not proportional to the Hodges numerator. More precisely:
\begin{enumerate}[label=\textup{(\roman*)}]
\item for $n=7$, one may take $\Delta_n=\Delta_{2,1}$ on six labels and a one-point completion as described in \Cref{sec:seven};
\item for $n=2k\geq8$, one may take $\Delta_n=\Delta_{k-1,1}$ on all $n$ labels;
\item for $n=2k+1\geq9$, one may take $\Delta_n=\Delta_{k-2,1}$ on the first $n-3$ labels.
\end{enumerate}
In each case, $P_n\in Q_n$ is chosen suitably.
Consequently, $\dim_\C(J_n)_{d(n)}\geq2$ for every $n\geq7$.
\end{theorem}

The proof is constructive.  Apart from the small cases $n=7$ and $n=9$, linear independence from the Hodges numerator is detected directly by a transposition: the Segre determinant and the correction monomial $P_n$ change sign under the transposition $(1\,2)$, so their product is invariant, whereas the Hodges numerator, coming from a determinant, is alternating.  For $n=7$ and $n=9$ we give points witnessing that the relevant polynomials are nonzero.

The paper is organized as follows.  In \Cref{sec:framework} we recall the spinor-helicity ring and develop the invariant theory of general Segre determinants in detail.  The correction monomial is introduced in \Cref{sec:psi}. The even and odd constructions are proved in \Cref{sec:even,sec:odd}, and the seven-point case is treated separately in \Cref{sec:seven}.  Witnessing points for $n=7,9$ are collected in the appendix.

\medskip
\noindent\textbf{AI Disclosure.}
The constructions in this paper were found using \texttt{ChatGPT Pro 5.6} and were subsequently verified and elaborated by the authors. More precisely, in the process of trying to prove the validity of \cite[Conjecture 2.1]{Gravity}, we came up with an induction procedure using representation theory. Then, we asked \texttt{ChatGPT Pro 5.6} to give a proof of this induction: however, it returned a counterexample to \cite[Conjecture 2.1]{Gravity} for $n=10$. Looking carefully at the counterexample, after feedback from Bernd Sturmfels, we realized that the central idea was the use of what we call Segre determinants in this article (which were studied recently by Pratt \cite{Segredeterminant}). In the counterexample, the AI tool was using an $8\times 8$ Segre determinant. We immediately realized that for $n=5,6$, Hodges numerators were actually Segre determinants. We generalized by ourselves this construction to any $n$ (Section \ref{sec:framework}). However, this construction by itself does not give the correct multidegrees. So \texttt{ChatGPT Pro 5.6} found for us monomials -- for both odd and even families -- that could close the gap with multidegrees and certify the memberships to the required ideals. We then checked and rewrote the proofs that these monomials indeed work (thus leading to write Sections \ref{sec:psi}, \ref{sec:even}, \ref{sec:odd}, \ref{sec:seven}). 

\medskip
\noindent\textbf{Acknowledgements.}
We thank Bernd Sturmfels and Jaroslav Trnka for useful discussions. E.V. and H.Y. are members of the GNSAGA group of INdAM. E.V. and H.Y. acknowledge that the preparation of this article was partially funded by the Italian Ministry of University and Research, through the project {\it Applied Algebraic Geometry of Tensors}, PRIN 2022 Protocol no. 2022NBN7TL.
J. K. is grateful for the support received from the Institute for Advanced Study and an anonymous donor.

\section{Spinor-helicity and Segre determinants}\label{sec:framework}

Let $L$ and $\widetilde L$ be two two-dimensional complex vector spaces and let
\[
 U=L^\vee,\qquad \widetilde U=\widetilde L^\vee.
\]
Thus $U$ and $\widetilde U$ are the spaces of linear coordinate functions on $L$ and $\widetilde L$, respectively.  Fix bases of $L$ and $\widetilde L$, and let $x,y\in U$ and $u,v\in\widetilde U$ be the corresponding dual bases.  On the $i$th copies $U_i\cong U$ and $\wt U_i\cong\wt U$, we write $x_i,y_i,u_i,v_i$ for the pullbacks of these coordinate functions.  A point is therefore written, in the chosen bases, as
\[
 \lambda_i=\binom{x_i}{y_i}\in L,\qquad
 \widetilde\lambda_i=\binom{u_i}{v_i}^{T}\in\widetilde L,
\]
where here $x_i,y_i,u_i,v_i$ denote the values of the corresponding coordinate functions at the given point.  Let the polynomial ring of functions on the spinor space be
\[
 S_n:=\C[L^{\oplus n}\oplus\widetilde L^{\oplus n}]
 =\Sym\bigl(U^{\oplus n}\oplus\widetilde U^{\oplus n}\bigr)
 \cong\C[x_1,y_1,\ldots,x_n,y_n,u_1,v_1,\ldots,u_n,v_n].
\]

We regard $p_{ij}=x_i y_j-y_i x_j=\det(\lambda_i,\lambda_j)$ and
$q_{ij}=u_i v_j-v_i u_j=\det(\widetilde\lambda_i,\widetilde\lambda_j)$ as variables.
The ring $S_n$ inherits a natural action of the algebraic group $\operatorname{SL}(U)\times\operatorname{SL}(\widetilde U)$.
The First Fundamental Theorem for $\operatorname{SL}(2,\C)$ \cite[Chapter 11, \S1.2, Theorem 1]{Procesibook} identifies the invariant subring of $S_n$ under this action with the tensor product of two Pl\"ucker rings
\[
 S_n^{\operatorname{SL}(U)\times\operatorname{SL}(\widetilde U)}
 \cong
 R_n:=\frac{\C[p_{ij},q_{ij}:1\leq i<j\leq n]}
 {(\text{$p$-Pl\"ucker relations, $q$-Pl\"ucker relations})}.
\]
This is the bihomogeneous coordinate ring of the product of two Grassmannians $\mathrm{Gr}(2,n)\times \mathrm{Gr}(2,n)$ in their Pl\"ucker embedding.
We shall use this identification repeatedly below.  Finally, set
\[
Q_n=R_n/(M_{ab}:1\leq a,b\leq n),\qquad
 M_{ab}=\sum_{r=1}^n p_{ar}q_{rb},
\]
where the $n^2$ relations $M_{ab}$ come from momentum-conservation relations of $n$ massless particles, each represented by the rank-one $2\times 2$ matrix $P_i = \lambda_i \widetilde{\lambda}_i$.
The ring $Q_n$ is the coordinate ring of the {\it spinor-helicity variety} $\mathrm{SH}(2,n,0)$
of El Maazouz, Pfister, and Sturmfels \cite{SH}.
This is the heart of the {\it spinor-helicity formalism}. Let $p_i$ be the complex four-momentum of a massless particle $i$. Then, in relativistic terms, the massless condition means that $p_i\cdot  p_i = p_i^2=0$, where the product is the Lorentzian inner product on $\C^4$; one can define an isomorphism between $\C^4$ and the space of $2\times 2$ complex matrices, given by $p_i\mapsto P_i$, such that the vanishing of the inner product $p_i^2=0$ is equivalent to the condition $\det(P_i)=0$.

\begin{remark}
    Let $V=\C^n$ with its standard basis $e_1,\ldots,e_n$ and its standard nondegenerate symmetric bilinear form.
    We look at the product of two Grassmannians $\mathrm{Gr}(2,n)\times \mathrm{Gr}(2,n)$. A point of the first Grassmannian is a $2$-dimensional plane $S\subset V$,  with Pl\"{u}cker coordinates $\pvar{ij}$. A point of the second Grassmannian is another $2$-dimensional plane $M\subset V$,  with Pl\"{u}cker coordinates $\qvar{ij}$.

    Choose $2\times n$ matrices $X$ and $Y$ whose row spans are $S$ and $M$.
    Then the momentum-conservation equations are expressed by the matrix equation $XY^T=0$.
    Equivalently, every vector in $S$ is orthogonal to every vector in $M$.
    Letting $W=M^\perp$, the condition $S\perp M$ becomes $S\subset W$.
    Hence a spinor-helicity point is identified with a two-step flag $S\subset W\subset V$, where $\dim S=2$ and $\dim W=n-2$.
    One can show that this identification yields an isomorphism of algebraic varieties between the spinor-helicity variety and the two-step flag variety $\mathbb{F}(2,n-2;n)$ in $V=\C^n$ \cite{SH}. In particular, $Q_n$ is an integral domain. Given $n$ spinors $\lambda_i$ and $\widetilde{\lambda}_i$, one has a pair of $2\times n$ matrices
    \[
     \Lambda:=(\lambda_1\ \cdots\ \lambda_n),
     \qquad
     \widetilde\Lambda:=(\widetilde\lambda_1^T\ \cdots\
     \widetilde\lambda_n^T).
    \]
    Then, each point of $\mathrm{SH}(2,n,0)$ can be represented by a pair of matrices $(\Lambda, \widetilde{\Lambda})$, where both have full rank and $\Lambda\widetilde{\Lambda}^T=0$. Thus, any such pair defines an evaluation homomorphism $Q_n\longrightarrow \C$ by taking the minors of $\Lambda$ and $\wt \Lambda$.
\end{remark}

For each pair $i<j$ put
\[
 J_{ij}=(p_{ij},q_{ij})\subset R_n,\qquad
 J_n=\bigcap_{i<j}\pi(J_{ij})\subset Q_n,
\]
where $\pi:R_n\to Q_n$ is the quotient map.  We use the \emph{spinor multidegree}
\[
 \deg p_{ij}=e_p+e_i+e_j\in \Z^{n+2},\qquad
 \deg q_{ij}=e_q-e_i-e_j\in \Z^{n+2},
\]
where $e_p, e_q, e_1,\ldots, e_n$ form the standard lattice basis of $\Z^{n+2}$. Then the ring $Q_n$ comes equipped with this spinor multidegree; for each $\alpha\in \Z^{n+2}$, $(Q_n)_{\alpha}$ is a finite-dimensional complex vector space.

\begin{remark}[Permutation symmetry of the Hodges numerator]
\label{rem:Hodges-alternating}
We shall use repeatedly that the Hodges numerator $N_n$ is alternating under
permutations of the labels. More precisely, for every
$\sigma\in\mathfrak{S}_n$,
\[
\sigma\cdot N_n=\operatorname{sgn}(\sigma)\,N_n.
\]
Equivalently, $N_n$ spans a copy of the sign representation of
$\mathfrak{S}_n$ inside $(J_n)_{d(n)}$. In particular, if a nonzero element
$F\in (J_n)_{d(n)}$ is invariant under some transposition, then
$F$ cannot be proportional to $N_n$.
\end{remark}

\subsection{Segre determinants}\label{subsec: Segre dets}

Fix integers $a,b\geq1$ and let $m=(a+1)(b+1)$.
The space of global sections of \(\cO_{\P(L)\times\P(\widetilde L)}(a,b)\) is
\[
 W_{a,b}:=H^0\!\left(\P(L)\times\P(\widetilde L),\cO(a,b)\right)
 \cong \Sym^a U\otimes\Sym^b\widetilde U,
\]
so $\dim W_{a,b}=m$.  The monomials
\[
 x^{a-r}y^r u^{b-s}v^s,
 \qquad 0\leq r\leq a,\quad 0\leq s\leq b,
\]
form a basis of $W_{a,b}$.  The associated complete linear system gives the {\it Segre--Veronese embedding}:
\[
 \nu_{a,b}:\P(L)\times\P(\widetilde L)\longrightarrow\P(W_{a,b}^{\vee})\cong\P^{m-1}.
\]
Let $I=\lbrace i_1,\ldots, i_m\rbrace\subseteq [n]$, so $|I|=m$.
Let $M_{a,b}(I)$ be the $m\times m$ matrix whose row indexed by $i\in I$ is given by the values $x_i^{a-r}y_i^r u_i^{b-s}v_i^s$.
We define
\[
 \Delta_{a,b}(I):=\det M_{a,b}(I)\in S_n.
\]
This is the {\it Segre determinant} associated with $\cO(a,b)$ in the sense of \cite{Segredeterminant}.

\begin{proposition}\label{prop:segre-basic}
    The polynomial $\Delta_{a,b}(I)$ is invariant under the action of
    $\operatorname{SL}(U)\times\operatorname{SL}(\widetilde U)$ on $S_n$.
    It defines a nonzero class in the Pl\"ucker ring $R_n$ of spinor multidegree
    \begin{equation}\label{eq:degree-general}
     \deg\Delta_{a,b}(I)
     =\left(\frac{am}{2},\frac{bm}{2};(a-b)^{m},0^{n-m}\right).
    \end{equation}
    Here the label coordinates are ordered with those in $I$ first.
    Moreover, for every pair of distinct indices $i,j\in I$,
    \begin{equation}\label{eq:segre-pair}
     \Delta_{a,b}(I)\in(p_{ij},q_{ij})\subset R_n.
    \end{equation}
    Finally, if $\sigma$ is a permutation of the labels in $I$, then
    \begin{equation}\label{eq:segre-sign}
     \sigma\cdot\Delta_{a,b}(I)=\operatorname{sgn}(\sigma)\Delta_{a,b}(I).
    \end{equation}
\end{proposition}

\begin{proof}
    We divide the proof into five steps.

    \smallskip
    \noindent\emph{Step 1: the group action.}
    Let $G:=\operatorname{SL}(U)\times\operatorname{SL}(\widetilde U)$.
    For each index $i$, the degree-one polynomial variables $x_i,y_i$ span a copy $U_i\simeq U$, while $u_i,v_i$ span a copy $\widetilde U_i\simeq\widetilde U$.  The group $G$ acts diagonally on
    $\bigoplus_{i=1}^n U_i\oplus\bigoplus_{i=1}^n\widetilde U_i$,
    and this action extends multiplicatively to the polynomial ring $S_n$.  In particular, $G$ acts on $W_{a,b}=\Sym^aU\otimes\Sym^b\widetilde U$.
    Let $\rho_{a,b}(g,h)$ be the matrix of this irreducible representation in the monomial basis
    \[
     x^{a-r}y^r u^{b-s}v^s,
     \qquad 0\leq r\leq a,\quad 0\leq s\leq b.
    \]
    The action on the monomials in each row of $M_{a,b}(I)$ is given by right multiplication by the transpose of $\rho_{a,b}(g,h)$.
    More succinctly,
    \[
        (g,h)\bullet M_{a,b}(I) = M_{a,b}(I) \left(\rho_{a,b}(g,h)\right)^T.
    \]
    It remains to compute the determinant
    of this matrix.

    For a two-dimensional complex vector space $E$ and $T\in\operatorname{GL}(E)$, one has $\det(\Sym^aT)=(\det T)^{a(a+1)/2}$.
    Indeed, if $T$ has eigenvalues $\alpha,\beta$, then $\Sym^aT$ has eigenvalues $\alpha^a,\alpha^{a-1}\beta,\ldots,\alpha\beta^{a-1},\beta^a$,
    whose product is $(\alpha\beta)^{a(a+1)/2}$.  The identity for arbitrary $T$ follows because both sides are polynomial characters on $\operatorname{GL}(E)$.  Hence, for $g\in\operatorname{SL}(U)$ and $h\in\operatorname{SL}(\widetilde U)$,
    $\det(\Sym^ag)=\det(\Sym^bh)=1$.
    Using the property of the determinant of the Kronecker product of two matrices, $\det(A\otimes B)=\det(A)^{\dim B}\det(B)^{\dim A}$,
    we obtain
    \[
     \det\rho_{a,b}(g,h)
     =\det(\Sym^ag)^{b+1}\det(\Sym^bh)^{a+1}=1.
    \]
    Thus $\Delta_{a,b}(I)$ is invariant under $G=\operatorname{SL}(U)\times\operatorname{SL}(\widetilde U)$.

    \smallskip
    \noindent\emph{Step 2: from spinor to Pl\"ucker coordinates.}
    Let
    \[
     A_I:=\Sym(U^{\oplus I}),\qquad
     B_I:=\Sym(\widetilde U^{\oplus I}).
    \]
    Then the polynomial ring on the supporting spinors is $S_I=A_I\otimes B_I$.  The first factor of $G$, i.e. $\operatorname{SL}(U)$ acts only on $A_I$, and the second factor $\operatorname{SL}(\widetilde U)$ acts only on $B_I$.  Taking invariants, we have $S_I^G
     =A_I^{\operatorname{SL}(U)}\otimes B_I^{\operatorname{SL}(\widetilde U)}$.
    The First Fundamental Theorem for $\operatorname{SL}(2,\C)$ states that $A_I^{\operatorname{SL}(U)}$ is generated by the $2\times2$ determinants $p_{ij}$, $i,j\in I$, and that the only relations among these generators are the Pl\"ucker relations; the same statement holds for the $q_{ij}$ on the second factor.  Thus, by \cite[Chapter 11, \S1.2, Theorem 1]{Procesibook},
    \begin{equation}\label{eq:FFT-product}
     S_I^G\cong
     \frac{\C[p_{ij}:i<j,\ i,j\in I]}{(p\text{-Pl\"ucker relations})}
     \otimes
     \frac{\C[q_{ij}:i<j,\ i,j\in I]}{(q\text{-Pl\"ucker relations})}.
    \end{equation}
    The right-hand side is the tensor product Pl\"ucker ring on the labels $I$, naturally embedded in $R_n$.  Since {\it Step 1} shows that $\Delta_{a,b}(I)\in S_I^G$, the determinant is a well-defined polynomial class in $R_n$.  We use the same symbol for this class.

    \smallskip
    \noindent\emph{Step 3: nonvanishing.}
    We prove that the determinant polynomial is not identically zero.  Choose pairwise distinct complex numbers
    $t_0,\ldots,t_a, \ell_0,\ldots,\ell_b$. Consider the $m=(a+1)(b+1)$ pairs $(t_r,\ell_s)$ with $0\le r\le a$ and $0\le s\le b$. For each such pair $(t_r,{\ell_s})$, define
    \[
     \lambda_{r,s}=(1,t_r)^T,\qquad
     \widetilde\lambda_{r,s}=(1,\ell_s).
    \]
    With columns ordered by the monomials $x^{a-r'}y^{r'}u^{b-s'}v^{s'}$, the corresponding Segre matrix is, up to simultaneous row and column reordering, the Kronecker product \linebreak $V_a(t_0,\ldots,t_a)\otimes V_b(\ell_0,\ldots,\ell_b)$, where $V_a$ and $V_b$ are the usual Vandermonde matrices.  Hence
    $\det M_{a,b}
     =\det(V_a)^{b+1}\det(V_b)^{a+1}\neq0$.
    Therefore $\Delta_{a,b}(I)$ is a nonzero polynomial in $S_I$, and, by the invariant-ring identification \eqref{eq:FFT-product}, a nonzero class in the product Pl\"ucker ring.

    \smallskip
    \noindent\emph{Step 4: the spinor multidegree.}
    The determinant is separately multihomogeneous in the spinors attached to every label $i$.  Indeed, every term in the determinant expansion uses exactly one entry from the row indexed by $i$.  Every entry in that row is homogeneous of degree $a$ in the coordinates $x_i,y_i$ of $\lambda_i$ and of degree $b$ in the coordinates $u_i,v_i$ of $\wt\lambda_i$.  Thus, in the polynomial ring $S_I$,
    \[
     \deg_{\lambda_i}\Delta_{a,b}=a,\qquad
     \deg_{\widetilde\lambda_i}\Delta_{a,b}=b
     \qquad(\mbox{ for every } i\in I),
    \]
    and both degrees are zero for labels outside $I$.

    Under the isomorphism \eqref{eq:FFT-product}, a factor $p_{ij}$ has degree one in the $L$-coordinates of label $i$ and one in those of label $j$, while a factor $q_{ij}$ has degree one in the $\widetilde L$-coordinates of labels $i$ and $j$.  Therefore the total number of $p$-factors of every monomial in $p_{ij},q_{ij}$ of an expansion of $\Delta_{a,b}$ in such variables is
    \[
     \deg_p\Delta_{a,b}=\frac{1}{2}\sum_{i\in I}a=\frac{am}{2},
    \]
    and similarly
    \[
     \deg_q\Delta_{a,b}=\frac{1}{2}\sum_{i\in I}b=\frac{bm}{2},
    \]
    where $m=|I|$.
    (These are integers because $a(a+1)$ and $b(b+1)$ are even.)  Finally, the $i$th label component of the spinor multidegree is, by definition, the number of $p$-occurrences of $i$ minus the number of $q$-occurrences of $i$.  Hence this component is $a-b$ for every $i\in I$, and $0$ outside $I$.  This proves \eqref{eq:degree-general}.

    \smallskip
    \noindent\emph{Step 5: ideal memberships.}
    Fix distinct $i,j\in I$ and work first in the polynomial ring $S_I$.  Consider the ideal
    \[
     K_{ij}:=(p_{ij},q_{ij})S_I.
    \]
    We claim that $K_{ij}$ is a prime ideal of $S_I$.  Geometrically, $p_{ij}=0$ is the irreducible determinantal hypersurface on which the two vectors $\lambda_i,\lambda_j$ are linearly dependent, and $q_{ij}=0$ is the analogous irreducible determinantal hypersurface for $\widetilde\lambda_i,\widetilde\lambda_j$; since the two equations involve disjoint sets of coordinates, their common zero set is the product of these two irreducible determinantal varieties with an affine space for the remaining labels, hence is irreducible.  Equivalently, $S_I/K_{ij}$ is a domain.

    On the zero locus $V(K_{ij})$, both pairs of spinors are linearly dependent.  Thus, up to exchanging $i$ and $j$, there are scalars $c,d$ such that $\lambda_j=c\lambda_i,
     \widetilde\lambda_j=d\widetilde\lambda_i$.
    The $j$th evaluation row of $M_{a,b}(I)$ is then $c^ad^b$ times the $i$th row. Hence $\Delta_{a,b}(I)$ vanishes identically on $V(K_{ij})$.  Since $K_{ij}$ is prime, Hilbert's Nullstellensatz yields
    $\Delta_{a,b}(I)\in K_{ij}$.
    Thus there exist $A,B\in S_I$ such that
    \begin{equation}\label{eq:pair-decomposition-polynomial}
     \Delta_{a,b}(I)=p_{ij}A+q_{ij}B.
    \end{equation}
    At this stage $A$ and $B$ need not be invariant. Now, the algebraic group $G=\operatorname{SL}(U)\times\operatorname{SL}(\widetilde U)$ is linearly reductive.  Hence there exists a Reynolds operator
    \[
     \mathcal R:S_I\longrightarrow S_I^G
    \]
    which is the identity on $S_I^G$ and is $S_I^G$-linear; see \cite[Chapters 8 and 14]{Procesibook}. (Such a Reynolds operator may be realized by integration over a maximal compact subgroup equipped with a suitable measure.)   The three polynomials $\Delta_{a,b}(I)$, $p_{ij}$ and $q_{ij}$ are $G$-invariant.  Applying $\mathcal R$ to \eqref{eq:pair-decomposition-polynomial} therefore gives
    \[
     \Delta_{a,b}(I)
     =p_{ij}\mathcal R(A)+q_{ij}\mathcal R(B),
    \]
    with $\mathcal R(A),\mathcal R(B)\in S_I^G$.  Via the First Fundamental Theorem isomorphism \eqref{eq:FFT-product}, these invariant coefficients are classes in the Pl\"ucker ring $R_n$.  Consequently $\Delta_{a,b}(I)\in(p_{ij},q_{ij})$ in $R_n$.  This proves \eqref{eq:segre-pair}.

    Finally, a permutation $\sigma$ of the labels in $I$ simply permutes the rows of the evaluation matrix $M_{a,b}(I)$.  Hence its determinant is multiplied by $\operatorname{sgn}(\sigma)$, which proves \eqref{eq:segre-sign} and completes the proof.
\end{proof}

\begin{remark}[The unique five- and six-point cases]\label{rem:five-six}
    Segre geometry already appears in the two multiplicities for which uniqueness was proved in \cite{Gravity}.

    For $n=5$, one has $d(5)=(2,2;0^5)$,
    which is the spinor multidegree of $\Delta_{1,1}$ on any four labels.  Let $M$ be the $5\times4$ matrix whose $i$th row is $(x_i u_i,\ x_i v_i,\ y_i u_i,\ y_i v_i)$.
    Momentum conservation $\Lambda\widetilde\Lambda^T=0$ is exactly the statement that the sum of the five rows $R_i$ of $M$ is zero.  For $r\in[5]$, let $D_r$ be the determinant of the $4\times4$ matrix obtained by deleting the $r$th row.  So, momentum conservation gives the row relation $R_1+R_2+R_3+R_4+R_5=0$. For $r<5$, substituting
    $R_5=-\sum_{j=1}^4R_j$ into the computation of the determinant $D_r$ implies that each determinant is $\pm D_5$.
    Thus, in $Q_5$, these determinants have, up to scaling, a common class, denoted $\Delta_5$.  By \Cref{prop:segre-basic}, $\Delta_5\in (p_{ij},q_{ij})$ and has degree $d(5)$.  It is easy to see that it must be nonzero. Five-point uniqueness from \cite{Gravity} then implies that $\Delta_5$ is proportional to the five-point Hodges numerator.  Thus the first two unique MHV numerators themselves admit a direct Segre-determinantal interpretation.

    For $n=6$, one has $d(6)=(6,3;1^6)$. Taking $(a,b)=(2,1)$ gives $m=6$, and \Cref{prop:segre-basic} yields $\deg\Delta_{2,1}(1,\ldots,6)=(6,3;1^6)=d(6)$.
    The same proposition shows that $\Delta_{2,1}$ belongs to every ideal $(p_{ij},q_{ij})$ with $1\leq i<j\leq6$.  Its image in $Q_6$ is nonzero.  Indeed, take $\lambda_i=(1,i-1)^T$ for $1\leq i\leq6$, and set
    \[
     \widetilde\lambda_1=(1,0),\quad
     \widetilde\lambda_2=(-2,0),\quad
     \widetilde\lambda_3=(1,0),
    \]
    \[
     \widetilde\lambda_4=(0,1),\quad
     \widetilde\lambda_5=(0,-2),\quad
     \widetilde\lambda_6=(0,1).
    \]
    One checks that $\Lambda\widetilde\Lambda^T=0$. After grouping the $u$- and $v$-columns, the Segre matrix is block diagonal with two nonzero $3\times3$ Vandermonde blocks.  Hence $\Delta_{2,1}\neq0$ in $Q_6$, and therefore
    $0\neq\Delta_{2,1}\in(J_6)_{d(6)}$.
    Since \cite{Gravity} proves that this space is one-dimensional, $\Delta_{2,1}$ is proportional to the six-point Hodges numerator.
\end{remark}

\section{A correction monomial}\label{sec:psi}

We now introduce a class of monomials $\Psi_r$ that we shall need later.

\begin{construction}
    The monomial $\Psi_r$ is defined via a graph.
    Fix $r\geq0$ and let $N=2r+6$ and $V=\{1,\ldots,N\}$.
    Let $A=\{1,\ldots,r+1\}$ and $B=V\setminus A$, so $|B|=r+5$.
    Then, let $C_r$ be the graph on $B$ constructed as follows.
    Identify the vertices in $B$ with elements of $\Z/(r+5)\Z$.

    \begin{enumerate}[
        leftmargin=6em,
        labelwidth=4em,
        labelsep=1em,
        align=left
    ]
        \item[\textbf{\boldmath $r$ even}]
        Then $|B|$ is odd, and $i,j\in\Z/|B|\Z$ form an edge
        of $C_r$ if and only if
        $j-i\equiv\pm1,\ldots,\pm r/2$ (modulo $|B|$).

        \item[\textbf{\boldmath $r$ odd}]
        Then $|B|$ is even, and $i,j\in\Z/|B|\Z$ form an edge
        of $C_r$ if and only if
        $j-i\equiv\pm1,\ldots,\pm(r-1)/2,\pm|B|/2$
        (modulo $|B|$).
    \end{enumerate}
\end{construction}

In both cases, the constructed graph $C_r$ is $r$-regular, and it is called a {\it circulant graph}.
Let $G_r=K_A\sqcup C_r$, where $K_A$ is the complete graph on the vertex set $A$.
Then $G_r$ is $r$-regular on $V$.  Let
\[
 \mathcal M_r=\{\{3,4\},\{5,6\},\ldots,\{2r+1,2r+2\}\}
\]
be a matching of size $r$ (empty if $r=0$), and define the monomial
\begin{equation}\label{eq:Psi}
 \Psi_r=
 \left(\prod_{\{i,j\}\in E(G_r)}p_{ij}\right)
 \left(\prod_{\{i,j\}\in\mathcal M_r}p_{ij}q_{ij}\right).
\end{equation}

\begin{lemma}\label{lem:Psi}
The spinor multidegree of $\Psi_r$ is
\begin{equation}\label{eq:Psi-degree}
 \deg\Psi_r=\bigl(r(r+4),r;r^N\bigr).
\end{equation}
Moreover, if $r\geq 1$,
\begin{equation}\label{eq:Psi-sign}
 (1\,2)\Psi_r=-\Psi_r.
\end{equation}
\end{lemma}

\begin{proof}
The $r$-regular graph $G_r$ has $Nr/2=(r+3)r$ edges.
For the total $p$-degree, each edge of $G_r$ contributes one $p_{ij}$; moreover, the matching contributes $r$ $p_{ij}$'s. Therefore the total $p$-degree is $(r+3)r+r=r(r+4)$. For the total $q$-degree, only the matching contributes $r$.

At an endpoint $i$ of the matching, the additional occurrences in the $p$- and $q$-factors cancel, so the difference between $p$-occurrences and $q$-occurrences is $r$ for each matched vertex. At a vertex outside the matching, there is no change, so there are only $r$ $p$-occurrences and zero $q$-occurrences, hence their difference is again $r$.   This proves the multidegree \eqref{eq:Psi-degree}.

When $r\geq 1$, the set of edges of $G_r$ is invariant under $(1\,2)$, because $1,2$ appear as vertices in the complete graph $K_A$. This graph contains the edge $\{1,2\}$, which contributes one sign because $p_{21}=-p_{12}$; all other edges incident to $1$ or $2$ are exchanged in pairs.  The circulant part and the matching avoid the labels $1,2$.  Hence \eqref{eq:Psi-sign}.
\end{proof}

\section{Even multiplicity}\label{sec:even}

In this section we establish our main result, Theorem~\ref{thm:main}, for the even case.

\medskip

Let $n=2k\geq8$ and define
$r=k-3$.
Since $W_{k-1,1}=\Sym^{k-1}U\otimes\widetilde U$ has dimension $2k=n$,
there is a Segre determinant $\Delta_{k-1,1}$ defined on all $n$ labels.  By \Cref{prop:segre-basic},
\[
 \deg\Delta_{k-1,1}
 =\bigl(k(k-1),k;(k-2)^n\bigr).
\]
Because $n=2r+6$, the polynomial $\Psi_r$ of \Cref{sec:psi} is defined on the same label set.

\begin{theorem}[Even family]\label{thm:even}
For $n=2k\geq8$, define
\begin{equation}\label{eq:Feven}
 F_{2k}=\Delta_{k-1,1}(1,\ldots,2k)\,\Psi_{k-3}.
\end{equation}
Then
\[
 0\neq F_{2k}\in(J_{2k})_{d(2k)}, \mbox{ and } (1\,2)F_{2k}=F_{2k}.
\]
In particular, $F_{2k}$ is not proportional to the Hodges numerator.
\end{theorem}

\begin{proof}
Using \Cref{lem:Psi}, the spinor multidegrees add to
\begin{align*}
\deg_pF_{2k}
 &=k(k-1)+(k-3)(k+1)=2k^2-3k-3,\\
\deg_qF_{2k}
 &=k+(k-3)=2k-3,
\end{align*}
and every label component is
$(k-2)+(k-3)=2k-5$.
This is precisely $d(2k)$.  Since the Segre determinant is supported on all labels, \eqref{eq:segre-pair} shows that it belongs to every ideal $(p_{ij},q_{ij})$.

We next check the non-vanishing of the Segre factor in $Q_{2k}$.  Split the labels into two sets $I_1,I_2$ of size $k$.  Choose $k$ distinct points $[\lambda_i]\in\P^1$ on each set, with nonzero representatives $\lambda_i\in L$.  The corresponding degree-$(k-1)$ Veronese evaluation matrices are invertible.  Since $k\geq4$, one may choose nonzero scalars $c_i$ on $I_1$ and $d_i$ on $I_2$ with
\[
 \sum_{i\in I_1}c_i\lambda_i=0,\qquad
 \sum_{i\in I_2}d_i\lambda_i=0.
\]

We check the previous statement. Since the points $[\lambda_i]\in\mathbb P^1$, $i\in I_1$, are
pairwise distinct, the $2\times k$ matrix having the $\lambda_i$ as
columns has rank $2$, and hence its kernel $K$ has dimension $k-2$.
Moreover, $K$ is not contained in any coordinate hyperplane
$\{c_i=0\}$. Indeed, fixing $i$ and choosing distinct
$j,\ell\neq i$, the vectors $\lambda_j,\lambda_\ell$ form a basis of
$\mathbb C^2$, so $\lambda_i=a\lambda_j+b\lambda_\ell$,
with $a,b\neq0$, giving a kernel vector whose $i$th coordinate is
nonzero. Thus each $K\cap\{c_i=0\}$ is a proper linear subspace of
$K$. Since $\mathbb C$ is infinite, a finite union of proper linear
subspaces cannot cover $K$. Hence there exists
$(c_i)_{i\in I_1}\in K$ with $c_i\neq0$ for every $i$, that is,
\[
\sum_{i\in I_1} c_i\lambda_i=0,
\qquad c_i\neq0\quad\text{for all }i\in I_1.
\]

Set $\widetilde\lambda_i=c_i(1,0)$ on $I_1$ and $\widetilde\lambda_i=d_i(0,1)$ on $I_2$.  Momentum conservation holds, while after grouping the columns according to their $u$- and $v$-factors the Segre matrix is block diagonal with two nonzero Veronese determinants.  Hence $\Delta_{k-1,1}$ is nonzero in $Q_{2k}$.  Since $Q_{2k}$ is a domain, each Pl\"ucker coordinate is a nonzero class, and $\Psi_r$ is a monomial in these variables, so $F_{2k}\neq0$.

Finally, both $\Delta_{k-1,1}$ and $\Psi_r$ change sign under $(1\,2)$, by \Cref{prop:segre-basic,lem:Psi}.  Their product is therefore invariant.  The Hodges numerator is alternating, so the nonzero $F_{2k}$ cannot be proportional to it.
\end{proof}

\section{Odd multiplicity at least nine}\label{sec:odd}

In this section we establish our main result, Theorem~\ref{thm:main}, for the odd case.

\medskip

Let $n=2k+1\geq9$ and set
$S=\{1,\ldots,2k-2\}, T=\{2k-1,2k,2k+1\}$.
Thus $|S|=n-3=2(k-1)$.  We use the Segre determinant
$\Delta_{k-2,1}(S)$,
whose spinor multidegree is
\begin{equation}\label{eq:odd-core-degree}
 \left((k-2)(k-1),k-1;(k-3)^{2k-2},0^3\right).
\end{equation}
Define the three-label completion
\begin{equation}\label{eq:Theta}
 \Theta_k=
 \left(\prod_{i\in S}\prod_{t\in T}p_{it}\right)
 \left(\prod_{\{u,v\}\subset T}q_{uv}\right).
\end{equation}
Its spinor multidegree is
\begin{equation}\label{eq:Theta-degree}
 \deg\Theta_k=\bigl(6k-6,3;3^{2k-2},(2k-4)^3\bigr).
\end{equation}
Indeed, each support label occurs in three crossing $p$-brackets, while each label of $T$ occurs in $2k-2$ crossing $p$-brackets and in two $q$-brackets.

Let $r=k-4\geq0$.
Since $|S|=2r+6$, the correction monomial $\Psi_r$ is naturally defined using only the labels in $S$.

\begin{theorem}[Odd family]\label{thm:odd}
For $n=2k+1\geq9$, define
\begin{equation}\label{eq:Fodd}
 F_{2k+1}=\Delta_{k-2,1}(S)\,\Psi_{k-4}(S)\,\Theta_k.
\end{equation}
Then
\[
 0\neq F_{2k+1}\in(J_{2k+1})_{d(2k+1)}.
\]
If $n\geq11$, then $(1\,2)F_{2k+1}=F_{2k+1}$, hence $F_{2k+1}$ is not proportional to the Hodges numerator.  The same conclusion holds for $n=9$.
\end{theorem}

\begin{proof}
The spinor multidegrees in \eqref{eq:odd-core-degree}, \eqref{eq:Psi-degree}, and \eqref{eq:Theta-degree} add as follows:
\begin{align*}
\deg_pF_{2k+1}
 &=(k-2)(k-1)+k(k-4)+(6k-6)=2k^2-k-4,\\
\deg_qF_{2k+1}
 &=(k-1)+(k-4)+3=2k-2.
\end{align*}
For $i\in S$, the label component is $(k-3)+(k-4)+3=2k-4$,
while for $t\in T$ it is already $2k-4$ from $\Theta_k$.  Thus $\deg F_{2k+1}=d(2k+1)$.

Pairs inside $S$ are covered by the Segre determinant.  Every pair joining $S$ to $T$ occurs as a $p$-factor in $\Theta_k$, and every pair inside $T$ occurs as a $q$-factor.  Hence $F_{2k+1}\in J_{2k+1}$.

We next prove that $F_{2k+1}$ is nonzero in $Q_{2k+1}$.
We first show that the Segre factor $\Delta_{k-2,1}(S)$
defines a nonzero class in $Q_{2k+1}$.
By \Cref{prop:segre-basic}, there exists a choice of the spinor pairs
$(\lambda_i,\widetilde\lambda_i), i\in S$,
for which $\Delta_{k-2,1}(S)\neq 0$.
These spinors need not satisfy momentum conservation.  Let $P_S:=\sum_{i\in S}\lambda_i\widetilde\lambda_i$
be their total momentum matrix.  This is an arbitrary $2\times2$ complex
matrix.  Since every $2\times2$ matrix is a sum of at most two rank-one
matrices, we may write
\[
-P_S=
\lambda_{2k-1}\widetilde\lambda_{2k-1}
+
\lambda_{2k}\widetilde\lambda_{2k}.
\]
If, say, $-P_S=
\begin{pmatrix}
a&b\\
c&d
\end{pmatrix},
$
one may take
\[
\lambda_{2k-1}=
\binom{a}{c},
\qquad
\widetilde\lambda_{2k-1}=
(1,0),
\qquad
\lambda_{2k}=
\binom{b}{d},
\qquad
\widetilde\lambda_{2k}=
(0,1).
\]
For the last label we may take, for instance, $\lambda_{2k+1}=0$,
and choose $\widetilde\lambda_{2k+1}$ arbitrarily.  Its contribution to
the momentum matrix is then zero.  Consequently,
$\sum_{i=1}^{2k+1}\lambda_i\widetilde\lambda_i
=
P_S+
\lambda_{2k-1}\widetilde\lambda_{2k-1}
+
\lambda_{2k}\widetilde\lambda_{2k}
=0$.
Thus the resulting $(2k+1)$-tuple satisfies the momentum-conservation
relations defining $Q_{2k+1}$.

Since $\Delta_{k-2,1}(S)$ depends only on the labels in $S$, its value at
this momentum-conserving configuration is the same nonzero value with
which we started.  Hence
$\Delta_{k-2,1}(S)\neq0$ in $Q_{2k+1}$.

Finally, $Q_{2k+1}$ is the homogeneous coordinate ring of the
spinor-helicity variety, equivalently of the two-step flag variety
described above, and in particular it is an integral domain.  Each
Pl\"ucker coordinate appearing in the monomials
$\Psi_{k-4}(S)$ and $\Theta_k$ is a nonzero class in $Q_{2k+1}$.
Therefore $\Psi_{k-4}(S)\Theta_k\neq0$ in
$Q_{2k+1}$, and hence $F_{2k+1}\neq0$.

If $k\geq5$, equivalently $n\geq11$, then $r=k-4\geq1$.  The Segre determinant and $\Psi_r$ are both anti-invariant under $(1\,2)$, while $\Theta_k$ is invariant.  Therefore $(1\,2)F_{2k+1}=F_{2k+1}$, so the element is not proportional to the Hodges numerator.

For $n=9$ one has $k=4$ and $\Psi_0=1$, so the preceding argument cannot be applied.  An exact momentum-conserving witness is recorded in \Cref{app:witnesses}; at that point  $F_9=0$ and  $((6\,7)F_9)\neq0$.
Hence $(6\,7)F_9$ is not equal to $-F_9$, and $F_9$ is not alternating.
\end{proof}

\section{The case \texorpdfstring{$n=7$}{n=7}}\label{sec:seven}

In this section we establish our main result, Theorem~\ref{thm:main}, for the case $n=7$, which plays a special role as it does not follow from the results established in the preceding sections.

\medskip

For $n=7$ the natural determinant is $\Delta_{2,1}(1,\ldots,6)$, whose degree is $(6,3;1^6,0)$.  The missing spinor multidegree is $(5,1;1^6,2)$.
For distinct $a,b\in\{1,\ldots,6\}$ define
\begin{equation}\label{eq:S7ab-new}
 \cS_7^{ab}
 =\sum_{k\in[6]\setminus\{a,b\}}
 q_{k7}p_{ak}p_{bk}
 \prod_{\ell\in[6]\setminus\{a,b,k\}}p_{\ell7}.
\end{equation}

\begin{lemma}\label{lem:S7-new}
The class of $\cS_7^{ab}$ in $Q_7$ is independent of $\{a,b\}$.  Denote the common class by $\cS_7$.  Moreover,
\[
 \cS_7\in(p_{i7},q_{i7})Q_7\qquad(1\leq i\leq6).
\]
\end{lemma}

\begin{proof}
For distinct $a,b,c\in[6]$, the Pl\"ucker relation
\[
 p_{bk}p_{c7}-p_{ck}p_{b7}=p_{bc}p_{k7}
\]
gives, after collecting the two terms $k=b,c$,
\[
 \cS_7^{ab}-\cS_7^{ac}
 =p_{bc}\!\prod_{\ell\in[6]\setminus\{a,b,c\}}p_{\ell7}
 \sum_{r\in[6]\setminus\{a\}}p_{ar}q_{r7}.
\]
The last sum is the momentum relation $M_{a7}$, hence the difference is zero in $Q_7$.  This proves independence of $\{a,b\}$.

For the pair $(i,7)$ choose a representative $\cS_7^{ab}$ with $i\notin\{a,b\}$.  Its $k=i$ summand contains $q_{i7}$; every other summand contains $p_{i7}$.  Thus the common class lies in $(p_{i7},q_{i7})$.
\end{proof}

\begin{theorem}[Seven-point family]\label{thm:seven}
The polynomial
\begin{equation}\label{eq:F7-new}
 F_7=\Delta_{2,1}(1,\ldots,6)\,\cS_7
\end{equation}
is a nonzero non-Hodges element of $(J_7)_{d(7)}$.
\end{theorem}

\begin{proof}
The degree is
\[
 (6,3;1^6,0)+(5,1;1^6,2)=(11,4;2^7)=d(7).
\]
Pairs inside $\{1,\ldots,6\}$ are covered by the Segre determinant, and the six remaining ideals $(p_{i7},q_{i7})$ are covered by \Cref{lem:S7-new}.  Hence $F_7\in(J_7)_{d(7)}$.

At the exact momentum-conserving point in \Cref{app:witnesses},
$F_7\neq0$ and $(6\,7)F_7\neq-F_7$.  The Hodges numerator is alternating, so $F_7$ is not proportional to it.
\end{proof}

\section{Conclusion}\label{sec:conclusion}

\Cref{thm:even,thm:odd,thm:seven} prove \Cref{thm:main}.  The construction shows that Segre--Veronese geometry provides a direct source of elements satisfying the MHV interpolation conditions.

Mai and Zhang's work \cite{MaiZhang} first established the failure of algebraic uniqueness at seven and eight points and analyzed the Bose-compatible sectors. The present note is in some sense complementary: we find an explicit Segre--Veronese construction producing algebraic solutions independent of the Hodges numerator at every multiplicity $n\geq7$.

\appendix
\section{Exact witnesses for \texorpdfstring{$n=7$ and $n=9$}{n=7 and n=9}}\label{app:witnesses}

\subsection{Seven points}

Let
\[
\Lambda=
\begin{bmatrix}
1&1&1&1&1&1&1\\
1&-2&-3&-1&2&3&0
\end{bmatrix}
\]
and
\[
\widetilde\Lambda=
\begin{bmatrix}
-2&5&-4&-2&2&-2&3\\
2&10&-4&-6&-6&4&0
\end{bmatrix}.
\]
Then $\Lambda\widetilde \Lambda^T=\mathbf 0$ and so momentum is conserved. A direct computation shows that $\Delta_{2,1}=-622\,080,
 \cS_7=-288,
 F_7=179\,159\,040$
and $(6\,7)F_7=821\,145\,600$.

\subsection{Nine points}

Let
\[
\Lambda=
\begin{bmatrix}
    0&4&1&-1&2&3&-2&3&2\\
    4&3&-2&-3&-1&0&1&0&1
\end{bmatrix}
\]
and
\[
\widetilde\Lambda=
\begin{bmatrix}
    4&1&-4&0&-3&-2&3&28&-33\\
    -1&-3&-1&-1&1&0&-4&-8&13
\end{bmatrix}.
\]
Then $\Lambda\widetilde \Lambda^T=\mathbf 0$ and so momentum is conserved.
For the nine-point polynomial \eqref{eq:Fodd}, a direct computation shows $F_9=0$ and $(6\,7)F_9\neq0$.
This is the witness used in \Cref{thm:odd}.

\printbibliography
\end{document}